\documentclass[
final,
nomarks
]{dmtcs-episciences}

\usepackage[utf8]{inputenc}
\usepackage{subfigure}
\usepackage[round]{natbib}

\usepackage{amsmath}
\usepackage{amsfonts}
\usepackage{amssymb}
\usepackage{latexsym}
\usepackage{amsthm}

\newtheorem{theorem}{Theorem}
\newtheorem{proposition}[theorem]{Proposition}

\newtheorem{corollary}[theorem]{Corollary}

\newcommand{\dist}{\mathrm{dist}}

\author{B\'ela Csaba \and Judit Nagy-Gy\"orgy\thanks{Supported by J\'anos Bolyai Research Scholarship of the Hungarian Academy of Sciences.}}
\title{On the Advice Complexity of Online Matching on the Line\thanks{This study was supported by the project TKP2021-NVA-09. Project no TKP2021-NVA-09 has been implemented with the support provided  by the Ministry of Culture and Innovation of Hungary from the National Research, Development and Innovation Fund, financed under the TKP2021-NVA funding scheme.}}
\affiliation{University of Szeged, Szeged, Hungary}

\begin{document}

\publicationdata{vol. 26:3}{2024}{21}{10.46298/dmtcs.14125}
{2024-08-28; 2024-08-28; 2024-11-08}{2024-11-18}

\maketitle

\begin{abstract}
   We consider the matching problem on the line with advice complexity. We give a 1-competitive online algorithm with advice complexity \begin{math}n-1\end{math}, and show that there is no 1-competitive online algorithm reading less than \begin{math}n-1\end{math} bits of advice. Moreover, for each \begin{math}0<k<n\end{math} we present a \begin{math}c(n/k)\end{math}-competitive online algorithm with advice complexity \begin{math}O(k(\log N + \log n))\end{math} where \begin{math}n\end{math} is the number of servers, \begin{math}N\end{math} is the distance of the minimal and maximal servers, and \begin{math}c(n)\end{math} is the complexity of the best online algorithm without advice.
\end{abstract}

\section{Introduction}

In the problem of online minimum matching (also known as online weighted matching), defined by \cite{KP} and \cite{KMV}, \begin{math}n\end{math} points of a metric space \begin{math}s_1,\ldots,s_n,\end{math} called \emph{servers}, are given. The \emph{requests} \begin{math}r_1,\ldots,r_n\end{math} are the elements of the metric space as well, arriving one by one according to their indices, and upon arrival each one has to be matched to an unmatched server at a cost equal to their distance.
The goal is to minimize the total cost.   
We say that \begin{math}n\end{math} is the size of the input.

In this paper, the metric space is the real line. We can assume without loss of the generality that \begin{math}s_1\le\ldots\le s_n\end{math}. Observe that each matching corresponds to a permutation.

Competitive analysis is commonly used to evaluate the performance of online algorithms. Here we consider a minimization problem. Algorithm \begin{math}\mathrm{A}\end{math} is \begin{math}c\end{math}-competitive if the cost of \begin{math}\mathrm{A}\end{math} is at most \begin{math}c\end{math} times larger than the optimal cost (plus some constant). For randomized algorithms, the expected value of the cost is compared to the optimum.
The competitive ratio of algorithm \begin{math}\mathrm{A}\end{math} is the minimal \begin{math}c\end{math} for which \begin{math}\mathrm{A}\end{math} is \begin{math}c\end{math}-competitive. Let us mention that this notion is also called \emph{weak competitive ratio} by some researchers, since it includes an additive constant.

There are no tight bounds for the competitive ratio on the line. 
\cite{ABNPS} have developed an \begin{math}O(n^{\log_2 (3+\epsilon)-1}/\epsilon)\end{math}-competitive deterministic algorithm. 
An \begin{math}O(\log n)\end{math}-competitive randomized algorithm is presented by \cite{GL}. \cite{PS} have given a lower bound \begin{math}\Omega(\sqrt{\log n})\end{math} for the competitive ratio of randomized algorithms. 

The term advice complexity for online algorithms was introduced by \cite{DKP}. In this paper we work in the \emph{tape model} introduced by \cite{BKKKM}. In this model, the online algorithm may read an infinite advice tape written by the oracle, and the advice complexity is the number of bits read.
The question is the following: \textit{How many bits of advice are necessary and sufficient to achieve a competitive ratio \begin{math}c\end{math}?} This includes determining the number of bits an optimal algorithm needs (in case \begin{math}c=1\end{math}). 
Advice complexity has been investigated for many online problems, see the survey paper \cite{BFKLM} for more information.

We note that the results of \cite{Mik} and \cite{PS} imply that no online algorithm with sublinear advice complexity can be \begin{math}O(1)\end{math}-competitive. 

In the next section we prove that the advice complexity of the best 1-competitive matching algorithm is \begin{math}n-1\end{math}. In Section~\ref{divide} we present a family of online matching algorithms. The algorithm, indexed by a positive integer \begin{math}k,\end{math} reads \begin{math}O(k(\log N + \log n))\end{math} bits of advice and has competitive ratio \begin{math}c(n/k)\end{math} if it uses a \begin{math}c(n)\end{math}-competitive matching algorithm as a subroutine. This shows how increasing the number of advice bits help to decrease the competitive ratio.

\section{The advice complexity of 1-competitive online algorithms}

In this section we give matching upper and lower bounds for the advice complexity of 1-competitive online algorithms.
Let us begin with two folklore results on the structure of any optimal matching on the line. For completeness we present the proofs.

\begin{proposition}\label{order}
	Consider an optimal matching corresponding to permutation \begin{math}\pi\end{math} (\textit{i.e.} \begin{math}r_i\end{math} is matched to \begin{math}s_{\pi(i)}\end{math}). 
	\begin{itemize}
		\item If \begin{math}r_i\le s_{\pi(j)}< s_{\pi(i)}\end{math} for some \begin{math}i,j\end{math} then \begin{math}r_j\le s_{\pi(j)}\end{math}. 
		\item If \begin{math}r_i\ge s_{\pi(j)}>s_{\pi(i)}\end{math} for some \begin{math}i,j\end{math} then \begin{math}r_j\ge s_{\pi(j)}\end{math}. 
	\end{itemize}
\end{proposition}

\begin{proof}
	Consider the case \begin{math}r_i\le s_{\pi(j)}< s_{\pi(i)}\end{math} (the other case is similar). Suppose, to the contrary, that \begin{math}r_j>s_{\pi(j)}\end{math}. If \begin{math}r_j\le s_{\pi(i)}\end{math} then
	\begin{eqnarray*}
		\dist(r_i,s_{\pi(j)})+\dist(r_j,s_{\pi(i)}) &=& s_{\pi(j)}-r_i+s_{\pi(i)}-r_j\\
		&<& s_{\pi(i)}-r_i+r_j-s_{\pi(j)}\\
		&=& \dist(r_i,s_{\pi(i)})+\dist(r_j,s_{\pi(j)})
	\end{eqnarray*}
	by assumption \begin{math}r_j>s_{\pi(j)}\end{math} but this is a contradiction. If \begin{math}r_j>s_{\pi(i)}\end{math} then
	\begin{eqnarray*}
		\dist(r_i,s_{\pi(j)})+\dist(r_j,s_{\pi(i)}) &=& s_{\pi(j)}-r_i+r_j-s_{\pi(i)}\\
		&<& s_{\pi(i)}-r_i+r_j-s_{\pi(j)}\\
		&=& \dist(r_i,s_{\pi(i)})+\dist(r_j,s_{\pi(j)})
	\end{eqnarray*}
	by assumption \begin{math}s_{\pi(i)}>s_{\pi(j)}\end{math} but this is a contradiction.
\end{proof}

\begin{proposition}\label{switch}
	Consider an optimal matching corresponding to permutation \begin{math}\pi\end{math}. Suppose that 
	\begin{displaymath}\max\{r_i,r_j\}\le\min\{s_{\pi(i)},s_{\pi(j)}\} \textrm{ or } \min\{r_i,r_j\}\ge\max\{s_{\pi(i)},s_{\pi(j)}\}.\end{displaymath} 
	Then the matching corresponding to \begin{math}\pi'\end{math} where \begin{math}\pi'(i)=\pi(j)\end{math}, \begin{math}\pi'(j)=\pi(i)\end{math} and \begin{math}\pi'(\ell)=\pi(\ell)\end{math} if \begin{math}\ell\ne i,j\end{math} is also optimal.
\end{proposition}

\begin{proof}
	Suppose that \begin{math}\max\{r_i,r_j\}\le\min\{s_{\pi(i)},s_{\pi(j)}\}\end{math} (the other case is similar). 
	\begin{eqnarray*}
		\dist(r_i,s_{\pi(i)})+\dist(r_j,s_{\pi(j)}) &=& s_{\pi(i)}-r_i+s_{\pi(j)}-r_j\\
		&=& s_{\pi'(j)}-r_j+s_{\pi'(i)}-r_i\\
		&=& \dist(r_j,s_{\pi'(j)})+\dist(r_i,s_{\pi'(i)}),
	\end{eqnarray*}
	therefore the sums of the distances of the matched points in the two matchings are equal.
\end{proof}

Consider an optimal matching corresponding to permutation \begin{math}\pi.\end{math} Let \begin{displaymath}L_\pi=\{r_i : s_{\pi(i)}\le r_i, 1\le i\le n\},\end{displaymath} the set of the requests matched to their left, and \begin{displaymath}R_\pi=\{r_i : s_{\pi(i)}> r_i, 1\le i\le n\},\end{displaymath} the set of the requests matched to their right. 
Algorithm \begin{math}\mathrm{LR}\end{math} serves \begin{math}r_i\end{math} in the following way:
\begin{itemize}
	\item if there is an unmatched server \begin{math}s_j\end{math} equal to \begin{math}r_i,\end{math} then match them;
	\item if all unmatched servers are greater than \begin{math}r_i,\end{math} then match \begin{math}r_i\end{math} to the least unmatched server;
	\item if all unmatched servers are less than \begin{math}r_i,\end{math} then match \begin{math}r_i\end{math} to the largest unmatched server;
	\item otherwise read a bit of advice:
	\begin{itemize}
		\item if it is 0 (it means that \begin{math}r_i\in L_\pi\end{math}), then match \begin{math}r_i\end{math} to the greatest unmatched server less than \begin{math}r_i\end{math};
		\item if it is 1 (it means that \begin{math}r_i\in R_\pi\end{math}) then match \begin{math}r_i\end{math} to the least unmatched server greater than \begin{math}r_i\end{math}.
	\end{itemize}
\end{itemize}

\begin{theorem}
	Algorithm \begin{math}\mathrm{LR}\end{math} reads at most \begin{math}n-1\end{math} bits of advice, and gives an optimal matching.    
\end{theorem}

\begin{proof}
	Note that algorithm \begin{math}\mathrm{LR}\end{math} reads at most one advice bit per request, and it does not read any for serving the last request. 
	
	We prove optimality by induction on \begin{math}n\end{math}. The case \begin{math}n=1\end{math} is trivial. Suppose that \begin{math}n>1\end{math} and the statement holds for all \begin{math}i<n\end{math}. Consider an optimal matching corresponding to permutation \begin{math}\pi\end{math}. If \begin{math}r_1\in L_\pi\end{math} and \begin{math}s_{\pi(1)}=\max\{s_j:s_j\le r_i, 1\le j\le n\},\end{math} then the statement is true by the induction hypothesis. Similarly, if \begin{math}r_1\in R_\pi\end{math} and \begin{math}s_{\pi(1)}=\min\{s_j:s_j>r_i, 1\le j\le n\},\end{math} then we are ready.
	
	Suppose now, that \begin{math}r_1\in L_\pi\end{math} and \begin{math}s_{\pi(1)}<\max\{s_j:s_j\le r_i, 1\le j\le n\}=s_{\pi(i)}\end{math} for some \begin{math}i>1\end{math}. Proposition~\ref{order} implies that \begin{math}r_i\in L_\pi\end{math} too. Therefore, by Proposition~\ref{switch}, there is an optimal matching corresponding to permutation \begin{math}\pi'\end{math} such that \begin{math}s_{\pi'(1)}=\max\{s_j:s_j\le r_i, 1\le j\le n\},\end{math} and we can apply the argument above. Similarly, if \begin{math}r_1\in R_\pi,\end{math} then there is an optimal matching corresponding to permutation \begin{math}\pi'\end{math} such that \begin{math}s_{\pi'(1)}=\min\{s_j:s_j\ge r_i, 1\le j\le n\}\end{math}, therefore the statement holds.
\end{proof}

The following theorem shows that there is no better optimal algorithm.

\begin{theorem}
	There is no online algorithm for the matching problem on the line which provides an optimal matching while reading less than \begin{math}n-1\end{math} bits of advice. 
\end{theorem}

\begin{proof}
	Consider an optimal algorithm \begin{math}\mathrm{A}\end{math}.
	For each \begin{math}n\in\mathbb{Z}^+\end{math} we construct a set \begin{math}\mathcal{I}_n\end{math} of inputs, each of size \begin{math}n,\end{math} in a recursive way. An input consists of a set of servers and a sequence of requests. The set of servers is the same for every input in \begin{math}I_n\end{math}: \begin{math}s_i=i\end{math} for each \begin{math}i\in \{1, \dots, n\}.\end{math} We only need to determine the sequences of requests in \begin{math}\mathcal{I}_n.\end{math} The first input \begin{math}\mathcal{I}_1\end{math} has one element consisting of one request: \begin{math}r_1=1\end{math}.
	
	Suppose, that \begin{math}n>1.\end{math} The request sequence \begin{math}\rho_{0},\end{math} consisting of the requests \begin{math}r_i=n-2^{-i}\end{math}, \begin{math}i=1,\ldots,n\end{math} plays a special role in \begin{math}\mathcal{I}_n.\end{math} It belongs to \begin{math}\mathcal{I}_n,\end{math} and every other 
	request sequence of \begin{math}\mathcal{I}_n\end{math} can be obtained from prefixes of \begin{math}\rho_0\end{math} by extending it to having length \begin{math}n\end{math} with request sequences from \begin{math}\mathcal{I}_k,\end{math} \begin{math}1\le k<n\end{math}.
	
	More formally, repeat the following for every 
	\begin{math}1\le k<n.\end{math} Take the prefix of \begin{math}\rho_0\end{math} of length \begin{math}n-k\end{math},
	denote it by \begin{math}\rho,\end{math} its elements are \begin{math}r_i=n-2^{-i}\end{math},
	\begin{math}i=1, \dots, n-k.\end{math} Then for every \begin{math}\rho'\in \mathcal{I}_k\end{math}
	we form a new element \begin{math}\rho\rho'\end{math} of \begin{math}\mathcal{I}_n.\end{math} 
	Hence, every prefix \begin{math}\rho\end{math} of \begin{math}\rho_0\end{math} with length \begin{math}n-k\end{math} 
	is extended into \begin{math}|\mathcal{I}_k|\end{math} different request sequences of \begin{math}\mathcal{I}_n\end{math}.
	It is easy to see, that \begin{math}s_n\end{math} is matched to \begin{math}r_{n-k}\end{math} in all optimal serving of the request sequence \begin{math}\rho\rho'\end{math}.
	Note also, that \begin{math}s_n=n\end{math} is matched to \begin{math}r_n\end{math} in every optimal serving of \begin{math}\rho_{0}.\end{math} 
	
	Using induction, one can easily show that
	\begin{displaymath}|\mathcal{I}_n|=1+\sum_{i=1}^{n-1}|\mathcal{I}_{i}|=2^{n-1}.\end{displaymath} 
	
	We prove that \begin{math}\mathrm{A}\end{math} needs different advice words for any two elements of \begin{math}\mathcal{I}_n\end{math}, and none of these advice words can be a prefix of the other. We proceed by induction on \begin{math}n\end{math}. The statement is trivial for \begin{math}n=1\end{math}. Suppose that \begin{math}n>1,\end{math} and the statement holds for all \begin{math}k<n\end{math}. Let \begin{math}\rho\end{math} consisting of \begin{math}r_1,\ldots,r_n\end{math} and \begin{math}\rho'\end{math} consisting of \begin{math}r'_1,\ldots,r'_n\end{math} two different elements of \begin{math}\mathcal{I}_n\end{math}. We have three cases.
	
	If \begin{math}\rho=\rho_0\end{math} and \begin{math}\rho'\end{math} is constructed for an element of \begin{math}\mathcal{I}_{k}\end{math} for some \begin{math}k<n\end{math}, then the first \begin{math}n-k\end{math} requests are identical, \textit{i.e.} \begin{math}r_1=r'_1,\ldots,r_{n-k}=r'_{n-k}\end{math}, moreover, in any optimal serving of \begin{math}\rho,\end{math} \begin{math}s_n\end{math} is matched to \begin{math}r_n,\end{math} and in any optimal serving of \begin{math}\rho',\end{math} \begin{math}s_n\end{math} is matched to \begin{math}r'_{n-k}\end{math}. Therefore \begin{math}\mathrm{A}\end{math} cannot distinguish between the two inputs knowing only the first \begin{math}n-k\end{math} requests, but the \begin{math}(n-k)\end{math}th requests require different treatment in the two inputs, therefore the statement holds.
	
	If \begin{math}\rho\end{math} is constructed for an element of \begin{math}\mathcal{I}_{k}\end{math} and \begin{math}\rho'\end{math} is constructed for an element of \begin{math}\mathcal{I}_{l}\end{math} for some \begin{math}k<l<n\end{math}, then the first \begin{math}n-l\end{math} requests are identical, \textit{i.e.} \begin{math}r_1=r'_1,\ldots,r_{n-l}=r'_{n-l}.\end{math} Moreover in any optimal serving of \begin{math}\rho,\end{math} \begin{math}s_n\end{math} is matched to \begin{math}r_{n-k}\end{math} and any optimal serving of \begin{math}\rho',\end{math} \begin{math}s_n\end{math} is matched to \begin{math}r'_{n-l}\end{math}. Therefore \begin{math}\mathrm{A}\end{math} cannot distinguish between the two inputs knowing only the first \begin{math}n-l\end{math} requests, but the \begin{math}(n-l)\end{math}th requests require different treatment in the two inputs, therefore the statement holds.
	
	If \begin{math}\rho\end{math} and \begin{math}\rho'\end{math} are constructed for two elements of \begin{math}\mathcal{I}_{k}\end{math} for some \begin{math}k<n\end{math}, then the first \begin{math}n-k\end{math} requests are identical, \textit{i.e.} \begin{math}r_1=r'_1,\ldots,r_{n-k}=r'_{n-k}\end{math}, moreover \begin{math}(n-n')\end{math}th request is matched to \begin{math}s_n\end{math} and the first \begin{math}n-k-1\end{math} requests are matched to \begin{math}s_{k+1},\ldots, s_{n-1}\end{math} in both cases, so the statement holds by the induction hypothesis.
\end{proof}

\section{Algorithms \begin{math}\mathrm{DIVIDE}_k\end{math} and \begin{math}\mathrm{RESCALE}\end{math}}\label{divide}

In this section we present an algorithm for each \begin{math}k>0\end{math} integer. We will see that there is a trade-off between the amount of advice and the competitive ratio: larger \begin{math}k\end{math} means a better competitive ratio but also more advice.  

First we assume that \begin{math}s_1=1\end{math}, \begin{math}s_n=N-1\end{math}, \begin{math}N\in\mathbb{Z}^+\end{math} and \begin{math}r_i\in\mathbb{Z}\end{math} for each \begin{math}i=1,\ldots,n\end{math}. Algorithm \begin{math}\mathrm{DIVIDE}_k\end{math} determines \begin{math}k\end{math} blocks on the line, each containing about \begin{math}n/k\end{math} servers, identifies requests whose pairs are outside, and works as a shell algorithm using algorithm \begin{math}\mathrm{LR}\end{math} to handle these requests, and a \begin{math}c(n)\end{math}-competitive algorithm \begin{math}\mathrm{A}\end{math} inside the blocks as subroutines.

Consider an optimal matching corresponding to permutation \begin{math}\pi\end{math} for which \begin{math}\pi(i)<\pi(j)\end{math} if \begin{math}r_i<r_j\end{math}. Note that it is well-known that such an optimal matching exists (it follows by Proposition~\ref{order} and Proposition~\ref{switch} as well).

A detailed description of Algorithm \begin{math}\mathrm{DIVIDE}_k\end{math} is as follows.

\vspace{10pt}

\underline{Pre-processing:} Let \begin{math}\ell \equiv n \mod k\end{math} (so \begin{math}0\le\ell<k\end{math}), and consider the following partitioning of the servers: 
\begin{displaymath}S_i=\{s_{(i-1)\lceil\frac{n}{k}\rceil+1},\ldots,s_{i\lceil\frac{n}{k}\rceil}\}\quad \textrm{ for every } 1\le i\le \ell,\end{displaymath}
\begin{displaymath}S_{\ell+j} = \{s_{\ell\lceil\frac{n}{k}\rceil+(j-1)\lfloor\frac{n}{k}\rfloor+1}, s_{\ell\lceil\frac{n}{k}\rceil+j\lfloor\frac{n}{k}\rfloor}\}\quad \textrm{ for every } 1\le j\le k-\ell-1.\end{displaymath}
Let
\begin{displaymath}p_i=\frac{1}{2}\left(\max\{s:s\in S_{i}\}+\min\{s:s\in S_{i+1}\}\right),\quad \textrm{ for every } i=1,\ldots, k-1.\end{displaymath}
Determining the blocks: \begin{displaymath}B_1=(-\infty, p_1],\ B_i=(p_{i-1},p_{i}],\ i=2,\ldots,k-1,\ B_k=(p_{k-1},\infty).\end{displaymath} 
Note that if \begin{math}p_{i-1}=p_i\end{math} then \begin{math}B_i=\emptyset\end{math}.\\[10pt]
Partitioning the points:
\begin{itemize}
	\item For each \begin{math}i=2,\ldots,k\end{math} the algorithm reads \begin{math}\lfloor\log_2 N\rfloor\end{math} bits of advice, \textit{i.e.}, a word \begin{math}q_{i,L}\end{math}, which is the minimum of \begin{math}N\end{math} and the position of the rightmost request of block \begin{math}B_i,\end{math} such that this request has to be served by a server in \begin{math}S_j\end{math} for some \begin{math}j<i\end{math}. If there is no such a request in block \begin{math}B_i,\end{math} then word \begin{math}q_{i,L}\end{math} is 0. If \begin{math}q_{i,L}= N\end{math} then set \begin{math}q_{i,L}:=\infty\end{math}. Let
	\begin{displaymath}L=\bigcup_{i=2}^k (p_{i-1},q_{i,L}],\end{displaymath} 
	where for \begin{math}a>b\end{math} we let  \begin{math}(a,b]=\emptyset.\end{math} 
	
	\item For each \begin{math}i=1,\ldots,k-1\end{math} the algorithm reads \begin{math}\lfloor\log_2 N\rfloor\end{math} bits of advice, \textit{i.e.}, a word \begin{math}q_{i,R}\end{math}, which is the maximum of 0 and the position of the leftmost request of block \begin{math}B_i,\end{math} such that this request has to be served by a server in \begin{math}S_j\end{math} for some \begin{math}j>i\end{math}. If there is no such a request in block \begin{math}B_i,\end{math} then word \begin{math}q_{i,R}\end{math} is the all-1 word of length \begin{math}\lfloor\log_2 N\rfloor.\end{math} If \begin{math}q_{i,R}= 0,\end{math} then set \begin{math}q_{i,R}:=-\infty\end{math}. Let
	\begin{displaymath}R=\bigcup_{i=1}^{k-1} [q_{i,R},p_{i}], \quad \textrm{ and } B=\mathbb{R}\setminus\{L\cup R\}.\end{displaymath} 
	Observe, that if a request \begin{math}r\end{math} belongs to \begin{math}B,\end{math} then its pair is in the block of \begin{math}r.\end{math} Note that there may be requests equal to some \begin{math}q_{i,L}\end{math} or \begin{math}q_{i,R}\end{math} with pairs in \begin{math}S_i\end{math}.
	
	\item (First marking procedure) Set \begin{math}M_R=\emptyset\end{math}. For each \begin{math}i=1,\ldots,k-1\end{math} if \begin{math}q_{i,R}>0,\end{math} then the algorithm reads 2 times \begin{math}\lfloor\log_2 n\rfloor\end{math} bits of advice, \textit{i.e.}, the number \begin{math}d_{i,R}\end{math} of requests equal to \begin{math}q_{i,R}\end{math} with pairs in \begin{math}S_i\end{math} and the number \begin{math}m=m_{i,R}\end{math} of requests in \begin{math}R\cap B_i\end{math} minus \begin{math}d_{i,R}\end{math}, \textit{i.e.}, the number of requests in \begin{math}B_i\end{math} matched outside to the right in the matching corresponding to \begin{math}\pi\end{math}. If \begin{math}m>0.\end{math} let \begin{math}s_j\end{math} be the server
	in \begin{math}\bigcup_{j>i}S_j\setminus M_R\end{math}
	with minimal index. 
	Set \begin{math}M_R:=M_R\cup\{s_j,\ldots,s_{j+m-1}\}\end{math}.
	
	\item (Second marking procedure) Set \begin{math}M_L=\emptyset\end{math}. For each \begin{math}i=0,\ldots,k-2\end{math} if \begin{math}q_{k-i,L}>0,\end{math} then the algorithm reads 2 times \begin{math}\lfloor\log_2 n\rfloor\end{math} bits of advice, \textit{i.e.}, the number \begin{math}d_{k-i,L}\end{math} of requests equal to \begin{math}q_{k-i,L}\end{math} with pairs in \begin{math}S_{k-i}\end{math} and the number \begin{math}m=m_{k-i,L}\end{math} of requests in \begin{math}L\cap B_{k-i}\end{math} minus \begin{math}d_{k-i,L}\end{math}, \textit{i.e.}, the number of requests in \begin{math}B_{k-i}\end{math} matched outside to the left in the matching corresponding to \begin{math}\pi\end{math}. If \begin{math}m>0,\end{math} let \begin{math}s_j\end{math} be the server  
	in \begin{math}\bigcup_{j<k-i}S_j\setminus M_L\end{math}
	with maximal index. 
	Set \begin{math}M_L:=M_L\cup\{s_j,\ldots,s_{j-m+1}\}\end{math}.
\end{itemize}
Let \begin{math}M=M_R\cup M_L\end{math} be the set of the marked servers and \begin{math}U\end{math} the set of the unmarked servers.
\vspace{10pt}

\underline{Servicing requests.} Set an empty auxiliary advice tape for algorithm \begin{math}\mathrm{LR}\end{math}. \\[10pt]
Handling request \begin{math}r_i\end{math}:
\begin{itemize}
	\item If \begin{math}r_i\in B\cap B_j\end{math} for some \begin{math}1\le j\le k,\end{math} then \begin{math}U:=U\cup\{r_i\}\end{math} and use algorithm \begin{math}\mathrm{A}\end{math} on inputs in \begin{math}U\cap (S_j\cup B_j)\end{math} to serve \begin{math}r_i\end{math}.
	\item If \begin{math}r_i=q_{j,R}\end{math} for some \begin{math}1\le j \le k-1\end{math} and the number of unmarked requests in \begin{math}U\end{math} equal to \begin{math}q_{i,R}\end{math} is less than \begin{math}d_{i,R},\end{math} then \begin{math}U:=U\cup\{r_i\},\end{math} and use algorithm \begin{math}\mathrm{A}\end{math} on inputs in \begin{math}U\cap (S_j\cup B_j)\end{math} to serve \begin{math}r_i\end{math}.
	\item If \begin{math}r_i=q_{j,L}\end{math} for some \begin{math}2\le j \le k\end{math} and the number of unmarked requests in \begin{math}U\end{math} equal to \begin{math}q_{i,L}\end{math} is less than \begin{math}d_{i,L},\end{math} then \begin{math}U:=U\cup\{r_i\},\end{math} and use algorithm \begin{math}\mathrm{A}\end{math} on inputs in \begin{math}U\cap (S_j\cup B_j)\end{math} to serve \begin{math}r_i\end{math}.
	\item If \begin{math}r_i\in R\cap B_j\end{math}, \begin{math}r_i>q_{j,R}\end{math} for some \begin{math}1\le j \le k-1\end{math} and the number of marked requests in \begin{math}B_j\end{math} is less than \begin{math}m_{j,R},\end{math} then mark \begin{math}r_i\end{math}, let \begin{math}M:=M\cup\{r_i\}\end{math}, write bit 1 at the end of the bit sequence on the auxiliary tape, and use algorithm \begin{math}\mathrm{LR}\end{math} on inputs in \begin{math}M\end{math} to serve \begin{math}r_i\end{math}. If algorithm \begin{math}\mathrm{A}\end{math} did not read any bits of advice to serve request \begin{math}r_i,\end{math} then delete the last bit of the content of the auxiliary tape.
	\item If \begin{math}r_i=q_{j,R}\end{math} for some \begin{math}1\le j \le k-1\end{math} and the number of unmarked requests equal to \begin{math}q_{i,R}\end{math} is \begin{math}d_{i,R},\end{math} then mark \begin{math}r_i\end{math}, let \begin{math}M:=M\cup\{r_i\}\end{math}, write bit 1 at the end of the bit sequence on the auxiliary tape and use algorithm \begin{math}\mathrm{LR}\end{math} on inputs in \begin{math}M\end{math} to serve \begin{math}r_i\end{math}. If algorithm \begin{math}\mathrm{A}\end{math} did not read any bits of advice to serve request \begin{math}r_i,\end{math} then delete the the last bit of the content of the auxiliary tape.
	\item If \begin{math}r_i\in L\cap B_j\end{math}, \begin{math}r_i<q_{j,L}\end{math} for some \begin{math}2\le j \le k\end{math} and the number of marked requests in \begin{math}B_j\end{math} is less than \begin{math}m_{j,L},\end{math} then mark \begin{math}r_i\end{math}, let \begin{math}M:=M\cup\{r_i\}\end{math}, write bit 0 at the end of the bit sequence on the auxiliary tape, and use algorithm \begin{math}\mathrm{LR}\end{math} on inputs in \begin{math}M\end{math} to serve \begin{math}r_i\end{math}. If algorithm \begin{math}\mathrm{A}\end{math} did not read any bits of advice to serve request \begin{math}r_i,\end{math} then delete the the last bit of the content of the auxiliary tape.
	\item If \begin{math}r_i=q_{j,L}\end{math} for some \begin{math}2\le j \le k\end{math} and the number of unmarked requests equal to \begin{math}q_{i,L}\end{math} is \begin{math}d_{i,L},\end{math} then mark \begin{math}r_i\end{math}, let \begin{math}M:=M\cup\{r_i\}\end{math}, write bit 0 at the end of the bit sequence on the auxiliary tape, and use algorithm \begin{math}\mathrm{LR}\end{math} on inputs in \begin{math}M\end{math} to serve \begin{math}r_i\end{math}. If algorithm \begin{math}\mathrm{A}\end{math} did not read any bits of advice to serve request \begin{math}r_i,\end{math} then delete the the last bit of the content of the auxiliary tape.
\end{itemize}

\vspace{10pt}

\begin{theorem}\label{main}
	Algorithm \begin{math}\mathrm{DIVIDE}_k\end{math} is \begin{math}c(n/k)\end{math}-competitive and reads \begin{math}O(k(\log N + \log n))\end{math} bits of advice.
\end{theorem}

\begin{proof}
	The second part of the statement follows immediately from the definition of algorithm \begin{math}\mathrm{DIVIDE}_k\end{math}.
	
	We can assume without loss of generality that whenever \begin{math}r_i=r_{i'}\end{math} and \begin{math}r_i\end{math} is matched to \begin{math}s_j\end{math} and \begin{math}r_{i'}\end{math} is matched to \begin{math}s_{j'}\end{math} by \begin{math}\mathrm{DIVIDE}_k\end{math} where \begin{math}j<j',\end{math} then \begin{math}\pi(i)<\pi(i')\end{math}. We will prove that if \begin{math}r_j\in B_i\cap U\end{math} then \begin{math}s_{\pi(j)}\in S_i\cap U\end{math}, moreover, if \begin{math}r_j\in M,\end{math} then \begin{math}s_{\pi(j)}\in M\end{math}.
	
	At first we need to see that the algorithm terminates. The number of servers and the number of requests are equal by definition of the model. The first and second marking procedure do not get stuck by definition of the matching corresponding to \begin{math}\pi\end{math} and the marking procedures. Moreover \begin{math}M_L\cap M_R=\emptyset\end{math} by definition of \begin{math}\pi\end{math}. Indeed, suppose, to the contrary that \begin{math}s_i\end{math} is marked in the first and the second marking procedure as well. There are requests \begin{math}r_j\le s_i\end{math} and \begin{math}r_\ell> s_i\end{math} such that \begin{math}\pi(j)\ge i\end{math} and \begin{math}\pi(\ell)\le i\end{math} by definition of \begin{math}\pi\end{math} and the pigeonhole principle, therefore \begin{math}\pi(j)>\pi(\ell)\end{math}, but this is a contradiction.
	
	The algorithm uses \begin{math}\mathrm{LR}\end{math} to match requests in \begin{math}M\end{math} to servers in \begin{math}M\end{math}, and it uses algorithm \begin{math}\mathrm{A}\end{math} to match unmarked requests in \begin{math}B_i\end{math} to unmarked servers in \begin{math}S_i\end{math} for each \begin{math}i\in\{1,\ldots,k\}\end{math}. By definition of algorithm \begin{math}\mathrm{DIVIDE}_k,\end{math} the number of the marked servers and the number of the marked requests are equal, and there is always an unread bit on the auxiliary advice tape whenever it is necessary.
	
	Moreover, the number of unmarked requests in \begin{math}B_i\end{math} and the number of unmarked servers in \begin{math}S_i\end{math} are equal, and the matching corresponding to \begin{math}\pi\end{math} match them to each other for all \begin{math}i\in\{1,\ldots,k\}\end{math} (this implies that a request in \begin{math}M\end{math} is matched to a server in \begin{math}M\end{math} in the optimal matching). We will prove this by induction on \begin{math}m=\sum_{i=1}^{k-1} (m_{i,R} + m_{i+1,L})\end{math}. It is easy to see that the statement holds for \begin{math}m=0\end{math}. 
	
	Now suppose that \begin{math}m>0,\end{math} and the statement holds for all smaller cases. Let \begin{math}i\end{math} be the smallest index with \begin{math}m_{i,R}>0\end{math} and \begin{math}j\end{math} the index of the first server in block \begin{math}B_{i+1}\end{math}. Then there is a request \begin{math}r_\ell\in B_i\cap R\end{math} matched to \begin{math}s_j\end{math} in the optimal matching, \textit{i.e.}, \begin{math}j=\pi(\ell)\end{math}, otherwise by the minimality of \begin{math}i\end{math} there is a request \begin{math}r_{\ell'}>p_i(\ge r_\ell)\end{math} such that \begin{math}j=\pi(\ell')\end{math} and \begin{math}j'>j\end{math} such that for the first request \begin{math}r_\ell\in B_i\cap R\end{math} and \begin{math}j'=\pi(\ell)\end{math}, but this is a contradiction by Proposition~\ref{order}. Deleting \begin{math}r_\ell\end{math} and \begin{math}s_j\end{math} from the model we are ready by the induction hypothesis. 
	
	If \begin{math}\sum_{i=1}^{k-1} m_{i,R} = 0,\end{math} then let \begin{math}i\end{math} be the largest index with \begin{math}m_{i,L}>0\end{math} and \begin{math}j\end{math} the index of the last server in block \begin{math}B_{i-1}\end{math}. Then there is a request \begin{math}r_\ell\in B_i\cap L\end{math} matched to \begin{math}s_j\end{math} in the optimal matching, \textit{i.e.}, \begin{math}j=\pi(\ell)\end{math}, otherwise, by the maximality of \begin{math}i,\end{math} there is a request \begin{math}r_{\ell'}\le p_i(<r_\ell)\end{math} such that \begin{math}j=\pi(\ell')\end{math} and \begin{math}j'<j\end{math} such that for the first request \begin{math}r_\ell\in B_i\cap L\end{math} and \begin{math}j'=\pi(\ell)\end{math}, but it is a contradiction by Proposition~\ref{order}. Again, deleting \begin{math}r_\ell\end{math} and \begin{math}s_j\end{math} from the model we are ready by the induction hypothesis.
	
	Since algorithm \begin{math}\mathrm{LR}\end{math} is optimal on \begin{math}M\end{math} and algorithm \begin{math}\mathrm{A}\end{math} is \begin{math}c(n/k)\end{math}-competitive on the blocks, we conclude that algorithm \begin{math}\mathrm{DIVIDE}_k\end{math} is \begin{math}c(n/k)\end{math}-competitive.
\end{proof}

\begin{corollary}
	If \begin{math}N=n^b\end{math} for some constant \begin{math}b,\end{math} then Algorithm \begin{math}\mathrm{DIVIDE}_k\end{math} is \begin{math}c(n/k)\end{math}-competitive and reads \begin{math}O(k\log n)\end{math} bits of advice.
\end{corollary}

Now we assume that the positions of requests and servers may be arbitrary real numbers. Algorithm \begin{math}\mathrm{RESCALE}\end{math} works in the following way:\\[5pt]
Fix \begin{math}k\in\mathbb{Z^+}\end{math}. Set \begin{math}s_i'=n^3 (s_i-s_1)+1\end{math}, \begin{math}r_i' = \lfloor n^3 (r_i-s_1)\rfloor+1\end{math} and \begin{math}N=\lceil s_n'-s_1'\rceil-2\end{math}. Execute algorithm \begin{math}\mathrm{DIVIDE}_k\end{math} on this modified input, and match \begin{math}r_i\end{math} to \begin{math}s_j\end{math} when \begin{math}\mathrm{DIVIDE}_k\end{math} matches \begin{math}r_i'\end{math} to \begin{math}s_j'\end{math}.

\begin{theorem}
	Algorithm \begin{math}\mathrm{RESCALE}\end{math} is \begin{math}c(n/k)\end{math}-competitive, and it reads \begin{math}O(k(\log (s_n-s_1) + \log n))\end{math} bits of advice.
\end{theorem}
\begin{proof}
	The number of bits read is \begin{math}O(k(\log N +\log n)) = O(k(\log (s_n-s_1) + \log n))\end{math} by the choice of \begin{math}N\end{math}.
	
	For bounding the optimal costs, observe that 
	\begin{displaymath}\mathrm{OPT}(I')\ge n^3\cdot \mathrm{OPT}(I) - 1/n^2,\end{displaymath}
	where \begin{math}I\end{math} is the original, \begin{math}I'\end{math} is the modified input.
	Moreover \begin{displaymath}\mathrm{RESCALE}(I')\le n^3\cdot \mathrm{DIVIDE}_k(I) + 1/n^2,\end{displaymath} 
	therefore, by Theorem~\ref{main} we have
	\begin{displaymath}\mathrm{RESCALE}(I')\le n^3\cdot c(n/k)\cdot\mathrm{OPT}_k(I) + \frac{1}{n^2}\le c(n/k)\cdot\mathrm{OPT}_k(I') + \frac{c(n/k)+1}{n^2},\end{displaymath} 
	where \begin{math}\frac{c(n/k)}{n^2}=o(1),\end{math} since we can assume that \begin{math}c(n)=o(n^2)\end{math}.
\end{proof}

 Using the algorithm of \cite{ABNPS} as algorithm \begin{math}\mathrm{A}\end{math} we get \begin{math}O((n/k)^{\log_2 (3+\epsilon)-1}/\epsilon)\end{math}-com\-pe\-ti\-ti\-ve algorithms.

\section{Conclusions and further questions} 

Note that \begin{math}\mathrm{A}\end{math} may be randomized. In this case we use the expected value of the cost of \begin{math}\mathrm{A}\end{math} instead of the cost of \begin{math}\mathrm{A}\end{math} in the definition of the competitive ratio. Applying the \begin{math}O(\log n)\end{math}-competitive randomized algorithm of \cite{GL} as algorithm \begin{math}\mathrm{A}\end{math} we get \begin{math}O(\log (n/k))\end{math}-competitive randomized algorithms in the previous section.

There is no known lower bound on the advice complexity of non-optimal online algorithms other than the result of \cite{Mik}, \textit{e.g.} for constant competitiveness greater than 1, and it is not known whether a linear number of bits of advice is sufficient for non-constant competitiveness.

The following research question arises: Can one modify the methods of the present paper for other metric spaces, \textit{e.g.} tree metric spaces? 

\nocite{*}
\bibliographystyle{abbrvnat}
\bibliography{matching}

\end{document}